\documentclass[reqno,11pt]{amsart}
\usepackage{amsfonts, amsthm}

\newtheorem{theorem}{Theorem}

\newtheorem{proposition}{Proposition}
\theoremstyle{definition}
\newtheorem{definition}{Definition}

\numberwithin{equation}{section}

\usepackage{hyperref}

\newcommand{\rme}{\mathrm{e}}

\begin{document}

\title[]{On relationships between symmetries depending on arbitrary functions and integrals of discrete equations}
\author[]{S. Ya. Startsev}

\urladdr{\href{http://www.researcherid.com/rid/D-1158-2009}{http://www.researcherid.com/rid/D-1158-2009}}

\address{Institute of Mathematics, Ufa Scientific Center, Russian Academy of Sciences}

\begin{abstract}
The paper is devoted to the conjecture that an equation is Darboux integrable if and only if it possesses symmetries depending on arbitrary functions. We note that results of previous works together prove this conjecture for scalar partial differential equations of the form $u_{xy}=F(x,y,u,u_x,u_y)$. For autonomous semi-discrete and discrete analogues of these equations we prove that the sequence of Laplace invariants is terminated by zero for an equation if this equation admits an operator mapping any function of one independent variable into a symmetry of the equation. The vanishing of an Laplace invariant allows us to construct a formal integral, i.e. an operator that maps symmetries into integrals (including, generally speaking, trivial integrals). This and results of previous works together prove a `formal' version of the aforementioned conjecture in the semi-discrete and pure discrete cases.    
\end{abstract}



\maketitle

\section{Introduction and the continuous case}\label{intro}
In the recent work \cite{Levi}, it was conjectured that the existence of symmetries depending on arbitrary functions is a necessary and \emph{sufficient} condition of the Darboux integrability for both partial differential equations and partial difference ones. Below we demonstrate that this conjecture is already proved for scalar partial differential equations of the form
\begin{equation}\label{hyp}
u_{xy}=F(x,y,u,u_x,u_y),
\end{equation}
and prove a formal version of the conjecture for differential-difference 
\[\left(u_{i+1}\right)_x = F\left(x,u_i,u_{i+1}, \left(u_i\right)_x\right), \qquad i \in \mathbb{Z}, \qquad \frac{\partial F}{\partial (u_i)_x} \ne 0,\]
and pure difference 
\[ u_{(i+1,j+1)} = F(u_{(i,j)},u_{(i+1,j)},u_{(i,j+1)}), \qquad i,j \in \mathbb{Z}, \qquad \frac{\partial F}{\partial u_{(i,j)}}\, \frac{\partial F}{\partial u_{(i+1,j)}}\, \frac{\partial F}{\partial u_{(i,j+1)}} \ne 0, \]
analogues of \eqref{hyp}. Here the semi-discrete equations are assumed to be uniquely solvable for $(u_i)_x$, and the discrete equations -- uniquely solvable for any argument of the right-hand side. 

Let us remind that equation \eqref{hyp} is said to be \emph{Darboux integrable} if it admits functions 
$w\left(x,y,u,\partial u /\partial x, \dots, \partial^n u /\partial x^n\right)$ and $\bar{w}\left(x,y,u,\partial u /\partial y, \dots, \partial^m u /\partial y^m\right)$ such that 
they essentially depend on at least one of the derivatives of $u$ and $D_y(w)=0$, $D_x(\bar{w})=0$; here $D_y$ and $D_x$ denote the total derivatives with respect to $y$ and $x$ by virtue of \eqref{hyp}, and the functions $w$ and $\bar{w}$ are called an $x$-\emph{integral} and a $y$-\emph{integral}, respectively. The definitions of Darboux integrability and integrals for the semi-discrete and discrete equations are almost the same, we only need to respectively replace partial derivatives of $u$ and the total derivatives by the shifts and the total differences for the corresponding discrete variables. 
In contrast to, for instance, \cite{HZS,GY}, the present paper deals only with the integrals without an explicit dependence on the discrete variables $i$ and $j$ as well as uses the existence of such integrals as a definition of Darboux integrability.\footnote{This assumption seems to be restrictive but it is very likely that Darboux integrable equations without an explicit dependence on $i$ and $j$ always admit integrals which do not depend on $i$ and $j$ too (although not all integrals of such equations are independent of $i$ and $j$ in accordance with examples in \cite{GY}).} 

According to \cite{ZhSS}, any Darboux integrable equation \eqref{hyp} admits operators of the form 
\begin{equation}\label{symopc}
S= \sum_{k=0}^{m-1} \alpha_k D_x^k, \qquad \bar{S}= \sum_{k=0}^{n-1} \bar{\alpha}_k D_y^k,
\end{equation}
such that $S(g)$, $\bar{S}(\bar{g})$ are symmetries of \eqref{hyp} for any $g \in \ker D_y$ and any $\bar{g} \in \ker D_x$. Here $g$, $\bar{g}$, $\alpha_k$ and $\bar{\alpha}_k$ may depend on $x$, $y$, $u$ and a finite number of the derivatives $\partial^p u /\partial x^p$, $\partial^q u /\partial y^q$ (all mixed derivatives of $u$ are excluded by virtue of \eqref{hyp}), and a function $f$ of the above variables is called a \emph{symmetry} if $(D_x D_y - F_{u_x} D_x - F_{u_y} D_y - F_u)(f)=0$. For brevity, the author offers to use the term `\emph{symmetry drivers}' for operators~\eqref{symopc} having the above properties. The most known example of the symmetry drivers was found in \cite{ZhSh} and is given by the operators $S= D_x + u_x$ and $\bar{S}= D_y+u_y$ for the Liouville equation $u_{xy}=\rme^u$. Note that arbitrary functions of $x$ and $y$ belong to $\ker D_y$ and $\ker D_x$, respectively, and symmetry drivers generate symmetries depending on arbitrary functions even if we do not assume the existence of integrals.

As it was shown in \cite{AdS}, the discrete and semi-discrete Darboux integrable equations also admit operators that map integrals and arbitrary functions of the independent variables ($i$, $j$ or $x$) into symmetries. The form of these symmetry drivers coincides with \eqref{symopc} up to replacing the derivatives with the shifts in $i$ and $j$. (More accurate definitions of the symmetry drivers are given below.) Thus, the part of the above conjecture is already proved for equation \eqref{hyp} and its aforementioned analogues: an equation admits symmetries depending on arbitrary functions if it is Darboux integrable. The present paper therefore focuses on converse statements.

For the continuous equations, such converse statement also was, in fact, proved in previous works. Indeed, according to \cite{SZh,AYu}, equation \eqref{hyp} is Darboux integrable if (and only if \cite{ZhSS,AK}) the sequence of the generalized Laplace invariants $h_k$, where $h_k$ are defined by the formulae
\[ h_{0} = F_u +F_{u_x} F_{u_y} - D_y(F_{u_y}), \qquad h_{1} = F_u +F_{u_x} F_{u_y} - D_x(F_{u_x}),   \]
\[ h_{k+1}= 2 h_k - D_x D_y(\ln(h_k))-h_{k-1}, \]
is terminated by zeros, i.e. if $h_p=0$ and $h_{-q}=0$ for some $p>0$ and $q \ge 0$. But the last condition holds for some $p \le m$ and $q < n$ if \eqref{hyp} admits symmetry drivers \eqref{symopc} (see \cite{StM}).
Hence, the Darboux integrability of \eqref{hyp} follows from the existence of symmetry drivers \eqref{symopc}.

In the present parer we prove that sequences of Laplace invariants for the semi-discrete and discrete analogues of \eqref{hyp} are also terminated by zeros if the corresponding equation admits symmetry drivers. In contrast to the continuous case, the termination of these sequences are not yet proved to be a sufficient condition of the Darboux integrability in the discrete and semi-discrete cases.
But the vanishing of Laplace invariants (together with other necessary conditions for the existence of symmetry drivers) allows us to construct formal integrals, i.e. operators that map symmetries into functions of integrals and one independent variable. Since the linearizations of integrals are formal integrals and the works \cite{ZhSS,AdS} in fact derive the symmetry drivers from these formal integrals, we obtain that \emph{the existence of symmetry drivers is necessary and sufficient condition for the existence of formal integrals}. It is noteworthy that, according to \cite{S17}, the last statement is valid for systems~\eqref{hyp} too (i.e. when $u$ and $F$ are vectors) despite the inapplicability of Laplace invariants to these systems. 
This statement therefore is quite general.

Thus, if the conjecture from \cite{Levi} is true in its original form, then the existence of formal integrals must be equivalent to the existence of `genuine' integrals. Possible ways to prove this are discussed at the ends of Sections~\ref{s2},\ref{deqs}. The connection between symmetry drivers and integrals is sometimes useful and, for example, allows to describe differential and difference substitutions of first order for function-parametrized families of evolution equations (see \cite{StM,foi,Stn}).   

\section{Differential-difference equations}\label{s2}

From now on, we, for brevity, omit $i$ in the subscripts of $u$ in all formulae and, in particular, write the aforementioned semi-discrete equation as
\begin{equation}\label{uix}
(u_1)_x = F(x,u,u_1,u_x).
\end{equation}
Due to the assumption $F_{u_x} \ne 0$, we can solve \eqref{uix} for $u_x$ and rewrite this equation in the form
\begin{equation}\label{umx}
(u_{-1})_x=\tilde{F}(x,u,u_{-1},u_x).
\end{equation}
The equations \eqref{uix}-\eqref{umx} allow us to express all derivatives $u_m^{(n)}:=\partial ^n u_{i+m} /\partial x^n$ of the shifts of $u$ in terms of $x$ and so-called \emph{dynamical variables} $u_p:=u_{i+p}$, $u^{(q)}:=\partial ^q u_i /\partial x^q$. The notation $g[u]$ indicates that the function $g$ depends on $x$ and a finite number of the dynamical variables. If a function $g$ may explicitly depend on $i$ in addition to $x$ and a finite set of the dynamical variables, then we use the notation $g_i[u]$. All functions are assumed to be analytical. 

Let $T$ denote the operator of the shift in $i$ by virtue of \eqref{uix}. It is defined by the following rules: $T(f_i(a,b,\dots))=f_{i+1}(T(a),T(b),\dots)$ for any function $f_i$; $T(u_p)=u_{p+1}$; $T(u^{(q)})=D^{q-1}(F)$ (i.e., the `mixed' variables $u_1^{(q)}$ are expressed in terms of $x$ and dynamical variables by using \eqref{uix}). Here $D$ is the total derivative with respect to $x$ by virtue of the equations \eqref{uix}-\eqref{umx}:
$$ D = \frac{\partial}{\partial x} + u^{(1)}\frac{\partial}{\partial u} + \sum_{k=1}^{\infty} \left( u^{(k+1)} \frac{\partial}{\partial u^{(k)}} + T^{k-1}(F) \frac{\partial}{\partial u_k} + T^{1-k}(\tilde{F}) \frac{\partial}{\partial u_{-k}}  \right),$$
The inverse shift operator $T^{-1}$ is defined in a similar way.

\begin{definition} A function $f_i[u]$ is called a \emph{symmetry} of equation \eqref{uix} if $L(f_i)=0$ for all $i$, where
\begin{equation}\label{lop}
L = T D - F_{u_x} D - F_{u_1} T - F_u.
\end{equation}
We say that operators $S=\sum_{k=0}^{\sigma} \alpha_k[u] D^k$ and $R=\sum_{k=0}^{r} \lambda_k[u] T^k$ are $x$- and $i$-\emph{symmetry drivers}, respectively, if $\alpha_{\sigma} \lambda_{r} \ne 0$, $\sigma, r \ge 0$ and $S(\xi)$, $R(\eta)$ are symmetries of \eqref{uix} for any $\xi_i[u] \in \ker (T-1)$ and any $\eta_i[u] \in \ker D$.  
\end{definition}
Since functions of $x$ and $i$ respectively belong to $\ker (T-1)$ and $\ker D$ for any equation \eqref{uix},
the above definition requires no assumptions about the existence of integrals (see Definition~\ref{sdint}). 

\begin{definition}\label{sdint} An equation of the form \eqref{uix} is called \emph{Darboux integrable} if there exist functions $X(x,u,u^{(1)},\dots,u^{(n)})$ and $I(x,u_\ell,u_{\ell+1},\dots,u_{\ell+m})$ such that $X_{u^{(n)}} \ne 0$, $I_{u_\ell} I_{u_{\ell+m}} \ne 0$ and the equalities $D(I)=0$, $T(X)=X$ hold. The functions $I$ and $X$ are respectively called an \emph{$i$-integral of order $m$} and an \emph{$x$-integral of order $n$} for the equation \eqref{uix}.

Operators $\mathcal{I}=\sum_{k=0}^{m} \mu_k[u] T^k$ and $\mathcal{X}=\sum_{k=0}^{n} \beta_k[u] D^k$ are said to be \emph{formal} $i$- and $x$-\emph{integrals}, respectively, if $\mu_m \ne 0$, $\beta_n \ne 0$, $m, n > 0$ and the operator identities $D \mathcal{I} = \sum_{k=0}^{m-1} \nu_k[u] T^k L$, $(T-1) \mathcal{X} = \sum_{k=0}^{n-1} \gamma_k[u] D^k L$ hold for some functions $\nu_k[u]$, $\gamma_k[u]$. 
\end{definition}
The last two defining relations mean that $\mathcal{I}$ and $\mathcal{X}$ map symmetries (if they exist) into $\ker D$ and $\ker (T-1)$. Since $T^{-\ell}$ maps $i$-integrals into $i$-integrals, we set $\ell = 0$ without loss of generality. Calculations similar to those used in \cite{Stn,StS} show that the \emph{linearizations} $I_{*}= \sum_{k=0}^m I_{u_k} T^k$,  $X_{*}= \sum_{k=0}^n X_{u^{(k)}} D^k$ of integrals $I$ and $X$ are formal $i$- and $x$-integrals, respectively. 

Any Darboux integrable equation \eqref{uix} admits both $x$- and $i$-symmetry drivers. This was proved in \cite{AdS} by using Laplace invariants. Let us define them. Introducing the main Laplace invariants $G_0 = F_u + F_{u_x} F_{u_1} - D(F_{u_x})$ and $H_0 = F_u + F_{u_x} T^{-1}(F_{u_1})$, we can rewrite \eqref{lop} as
\begin{equation}\label{fact}
L=(D-F_{u_1})(T-F_{u_x})-G_0 = (T-F_{u_x})(D-T^{-1}(F_{u_1}))-H_0.
\end{equation} 
If $G_0 \ne 0$, then we set $a_1 = F_{u_x} T(G_0)/G_0$, $L_{-1} = (T-a_1)(D-F_{u_1}) - T(G_0)$. The direct check shows that $(T-a_1) L = L_{-1} (T - F_{u_x})$. The transition from $L_0:=L$ to $L_{-1}$ is called the Laplace $i$-transformation. It can be applied to any operator of the form $T D - a[u] D - b[u]T - c[u]$ if we replace $F_{u_x}$, $F_{u_1}$ and $F_u$ with $a$, $b$ and $c$ in the above formulae. In particular, we can rewrite  
\[ L_{-1}=(D-T(F_{u_1}))(T-a_1)-G_1, \qquad G_1= T(G_0)- D(a_1) + a_1 (T(F_{u_1}) -F_{u_1}) \]
and apply the Laplace $i$-transformation to $L_{-1}$ if $G_1 \ne 0$, and so on. Repeating this procedure,
we obtain the sequence of the operators $L_{-k}=(D-T^k(F_{u_1}))(T-a_k)-G_k$, $k > 0$, where $a_k$ and the \emph{Laplace $i$-invariants} $G_k$ are defined by the recurrent formulae
\[ a_k=a_{k-1} T(G_{k-1})/G_{k-1},\quad a_0=F_{u_x}, \qquad G_k= T(G_{k-1})- D(a_k) + a_k (T^k(F_{u_1}) -T^{k-1}(F_{u_1})). \]

The Laplace $x$-transformation is defined in a similar way. The iterations of this transformation generate the sequence of the operators $L_k = (T-F_{u_x})(D-T^{-1}(b_k))-H_k$, $k > 0$,
where $b_k$ and the \emph{Laplace $x$-invariants} $H_k$ are calculated by the formulae
\[ b_k = T^{-1}(b_{k-1}) + D(H_{k-1})/H_{k-1}, \quad b_0=F_{u_1}, \qquad
H_{k}=H_{k-1}+ F_{u_x} ( T^{-1}(b_k) - b_k) + D(F_{u_x}), \]
By construction, the Laplace invariants and the operators $L_k$ for $k>0$ satisfy the equalities
\begin{eqnarray}
& (T - a_k)L_{1-k}=L_{-k}(T-a_{k-1}), \qquad (T - a_k)  G_{k-1} = T(G_{k-1})  (T-a_{k-1}), \label{il} \\
& (D - b_k)L_{k-1}=L_{k}(D-T^{-1}(b_{k-1})). \label{ild}
\end{eqnarray}
Here and below we use the notation $P g$ for the composition of an operator $P$ and the multiplication by a function $g[u]$, i.e. $P g$ is an operator and differs from the function $P(g)$.

The following statements was proved in \cite{AdS}:

\noindent 1) if equation \eqref{uix} admits an $x$-integral $X$ of order $n$, then
\begin{equation}\label{condx}
H_q=0 \;\; \text{for some} \; q<n \qquad \quad \text{and} \qquad \quad \exists \theta[u] \ne 0\;\;\; \text{such that} \;\; T(\theta)=F_{u_x}^{-1} \theta; 
\end{equation}

\noindent 2) if equation \eqref{uix} admits an $i$-integral $I$ of order $m$, then
\begin{equation}\label{condi}
G_p=0 \;\; \text{for some} \; p<m \qquad \quad \text{and} \qquad \quad \exists \tau[u] \ne 0\;\;\; \text{such that} \;\; D(\tau) + F_{u_1} \tau=0;
\end{equation}

\noindent 3) if both conditions \eqref{condx} and \eqref{condi} hold, then
\begin{equation}\label{xsd}
 S= \frac{1}{G_{0}} \left(D - F_{u_1} \right) \dots \frac{1}{G_{p-1}} \left(D - T^{p-1}(F_{u_1}) \right)   \frac{G_0 \dots G_{p-1}}{\theta}, 
\end{equation}
\begin{equation}\label{isd}
R= \frac{1}{H_{0}} \left(T - F_{u_x} \right) \dots \frac{1}{H_{q-1}} \left(T - F_{u_x} \right)  \frac{T^{-q}(H_0) \dots T^{-1}(H_{q-1})}{T^{-(q+1)}(\tau)} 
\end{equation}
respectively are $x$- and $i$-symmetry drivers of \eqref{uix}. ($S=\theta^{-1}$ and $R=T^{-1}(\tau^{-1})$ in the cases $p=0$ and $q=0$, respectively.)

But the corresponding proofs use only the fact that $X_*$ and $I_*$ are formal integrals (see a similar reasoning in the proof of Proposition~\ref{p2} below), and the work \cite{AdS} actually proves the necessity of the conditions~\eqref{condx} and \eqref{condi} for the existence of \emph{formal} $x$- and $i$-integrals, respectively. The converse statements can easily be derived from \eqref{il},\eqref{ild}: the direct calculation\footnote{See \cite{foi} or a similar calculation in the proof of Proposition~\ref{t2} below if more details are needed.} shows that
\begin{equation}\label{xfi}
\mathcal{X}=\theta[u] (D-T^{-1}(b_{q})) (D-T^{-1}(b_{q-1})) \dots (D-T^{-1}(b_0))
\end{equation}
is a formal $x$-integral of equation \eqref{uix} if \eqref{condx} holds, while \eqref{condi} implies that
\begin{equation}\label{ifi}
\mathcal{I}=T^p(\tau) (T-a_p)(T-a_{p-1})\dots (T-a_0)
\end{equation}
is a formal $i$-integral. In addition, the work \cite{foi} proves that $H_q=0$ for some $q \le r$ and $D(\tau) + F_{u_1} \tau =0$ for $\tau = 1/T(\lambda_r)$ if \eqref{uix} admits an $i$-symmetry drivers $R=\sum_{k=0}^{r} \lambda_k[u] T^k$.

\begin{theorem}\label{msd}
Equation~\eqref{uix} admits both formal $i$-integrals and formal $x$-integrals if and only if it possesses both $i$- and $x$-symmetry drivers.
\end{theorem}
Taking the previous two paragraphs into account, we only need to prove the following statement for establishing Theorem~\ref{msd}.
\begin{proposition} If equation \eqref{uix} admits an $x$-symmetry driver $S=\sum_{k=0}^{\sigma} \alpha_k[u] D^k$, then $G_p=0$ for some $p \le \sigma$ and $T(\theta)=F_{u_x}^{-1} \theta$ holds for $\theta = \alpha_{\sigma}^{-1}$.
\end{proposition}
\begin{proof}
Collecting the coefficients at $f^{(k)}$, $k=\overline{0,\sigma+1}$, in the equality $L(S(f(x))) = 0$ and taking the arbitrariness of $f(x)$ into account, we obtain the following chain of the relations:
\begin{eqnarray}
&&\left(T - F_{u_x}\right)(\alpha_\sigma) = 0, \label{ch1} \\
&&\left(T - F_{u_x}\right)(\alpha_{k-1}) + L(\alpha_k) = 0,\quad  1 \le k \le \sigma,  \label{ch2} \\ &&L(\alpha_0) = 0. \label{ch3}
\end{eqnarray}
Introducing $\alpha_{-1}=0$, we consider \eqref{ch3} as an extension of \eqref{ch2} for the case $k=0$.

Let $A_{-1}$ and $\bar{A}_0$ denote the identity mapping and the operators $A_{k}$ and $\bar{A}_q$ be defined by the recurrent formulae $A_k=(T-a_k) A_{k-1}$, $k \ge 0$, $\bar{A}_q=(T-a_q) \bar{A}_{q-1}$, $q > 0$. If $G_p \ne 0$ for all $p \le \sigma$, then we can prove the equalities
\begin{equation}\label{ind}
G_p A_{p-1}(\alpha_{\sigma - p}) = A_p(\alpha_{\sigma - p -1}) , \qquad A_p(\alpha_{\sigma - p})=0 
\end{equation}
by induction on $p$. Indeed, \eqref{ind} for $p=0$ follows from \eqref{ch1}, \eqref{ch2} for $k=\sigma$ and \eqref{fact}. If \eqref{ind} holds for some $p < \sigma$, then we obtain $A_{p+1}(\alpha_{\sigma - p-1})=0$ by applying $T-a_{p+1}$ to the first equation of \eqref{ind} and taking the second equations of \eqref{il}, \eqref{ind} into account. Since $\bar{A}_{p+1} L = L_{-(p+1)} A_p = (D-T^{p+1}(F_{u_1})) A_{p+1} - G_{p+1} A_p$ by \eqref{il}, the application of $\bar{A}_{p+1}$ to \eqref{ch2} for $k=\sigma - p - 1$ gives rise to the first equation of \eqref{ind} for $p+1$. 

Thus, \eqref{ind} is valid for all $p \le \sigma$ if we assume $G_p \ne 0$ for all $p \le \sigma$. In particular, $G_\sigma A_{\sigma-1}(\alpha_{0}) = 0$. The last equation and the first equality of \eqref{ind} (used as a recurrent formula) imply $A_{p-1}(\alpha_{\sigma - p})=0$ for all $p \le \sigma$. But this contradicts the condition $A_{-1}(\alpha_\sigma) = \alpha_\sigma \ne 0$.
\end{proof}

If a formal integral of \eqref{uix} is known, then we can try to obtain a `genuine' integral by applying the formal one to a symmetry of \eqref{uix}. (Theorem~\ref{msd} guarantees the existence of symmetries if formal integrals exist.) As an illustrative example of this, let us consider the equation 
\begin{equation}\label{sdliuv}
 \left(u_1\right)_x = u_x + \rme^{u}+\rme^{u_1}
\end{equation}
from \cite{AdS}. Since $H_1=G_1=0$, $\theta=1$ and $\tau=\rme^u/(\rme^{u_1}+\rme^u)$ for this equation, we can use \eqref{xfi}-\eqref{ifi} to construct its formal integrals
\begin{equation}\label{dlfi}
\mathcal{X}= D^2 - u_x D - \rme^{2u}, \quad \mathcal{I}= \frac{\rme^{u_1}}{\rme^{u_2}+\rme^{u_1}} T^2 - \left( \frac{\rme^{u_1}}{\rme^{u_1}+\rme^u} + \frac{\rme^{u_1}}{\rme^{u_2}+\rme^{u_1}} \right) T + \frac{\rme^{u_1}}{\rme^{u_1}+\rme^u}.
\end{equation}
The application of $\mathcal{X}$ to the symmetry $u_x$ gives us the $x$-integral $u_{xxx} - u_x u_{xx} -\rme^{2u} u_x$. As shown in \cite{foi}, this method allows us to prove the existence of integrals for entire subclasses of~\eqref{uix}. 

To obtain integrals, we can also employ the following way (which was used for equations \eqref{hyp} in \cite{ZhibSok}). Let $R$ be an $i$-symmetry driver and $\mathcal{I}$ be a formal $i$-integral of \eqref{uix}. Then the composition $\mathcal{I} R$ can be rewritten as $\sum_{k=0}^{m+r} w_k[u] T^k$ and maps $\ker D$ into $\ker D$ again. But this is possible only if $w_k[u] \in \ker D$. The same is true for $x$-symmetry drivers and formal $x$-integrals: coefficients of their compositions belong to $\ker (T-1)$. In the case of \eqref{sdliuv}, formulae \eqref{xsd}-\eqref{isd} gives us the symmetry drivers $S=D+u_x$, $R=(\rme^{u-u_{-1}}+1)T - \rme^{u_{-1}-u}(\rme^{u_{-1}-u_{-2}}+1)$, and
\[\mathcal{X} S=D^3+X D + \frac{D(X)}{2}, \qquad \mathcal{I} R = T^3 + (1-I) T^2 + T^{-1}(I) \left(\frac{1}{I} - 1 \right) T - \frac{T^{-2}(I)}{T^{-1}(I)}, \] 
where $X=2u_{xx}-u_x^2-\rme^{2u}$ and $I=(1+\rme^{u_1-u_2})(1+\rme^{u_1-u})$ are integrals of smallest order. 

The compositions of symmetry drivers and formal integrals are independent of the dynamical variables and generate no integrals for some equations~\eqref{uix}. But such independence for the compositions of \eqref{isd} with \eqref{ifi} and \eqref{xsd} with \eqref{xfi} seems to be a fairly restrictive condition and, likely, guarantees the existence of an additional symmetry which is mapped into `genuine' integrals by formal ones. For example, the coefficients of both formal integrals \eqref{xfi}-\eqref{ifi} and symmetry drivers \eqref{xsd}-\eqref{isd} depend on $x$ only for any equation  $(u_1)_x =a(x) u_x + b(x) u_1 + c(x) u$ such that $H_q=G_p=0$, but $u$ is a symmetry of this equation and the formal integrals map this symmetry into functions essentially depending on dynamical variables (i.e. into `genuine' integrals). Thus, it is probably that 
the methods of the previous two paragraphs complement each other and at least one of them can give us `genuine' integrals in any situation.

The proofs of the conditions \eqref{condx}-\eqref{condi} in \cite{AdS} also guarantee that the linearizations of $x$- and $i$-integrals are defined by formulae \eqref{xfi}-\eqref{ifi} if the orders of these integrals are $q+1$ and $p+1$, respectively (see the proof of Proposition~\ref{p2}). The integrals of smallest orders have orders $q+1$ and $p+1$ in all examples known to the author. This gives us an additional `heuristic' way to find integrals via \eqref{xfi}-\eqref{ifi}. For example, $X_*= 2 \mathcal{X}$ and $(\ln I)_*=-\mathcal{I}$, where $\mathcal{X}$, $\mathcal{I}$ are defined by \eqref{dlfi} and $X$, $I$ are the aforementioned integrals of smallest orders for \eqref{sdliuv}. It should be noted that $\ker (T - F_{u_x}^{-1})$ and $\ker (D + F_{u_1})$ are closed under multiplication by $x$- and $i$-integrals, respectively. Therefore, $\theta$ and $\tau$ in \eqref{xfi}-\eqref{ifi} are not uniquely defined. We obviously need to select $\theta$ and $\tau$ so that they do not depend on arguments other than arguments of integrals if we assume that \eqref{xfi}-\eqref{ifi} coincide with the linearizations of these integrals.

\section{Quad-graph equations}\label{deqs}
Let us introduce the notation $u_{p,q}:=u_{(i+p,j+q)}$, $u:=u_{0,0}=u_{(i,j)}$. According to it, the aforementioned discrete equation reads
\begin{equation}\label{uij}
u_{1,1}=F(u,u_{1,0},u_{0,1}).  
\end{equation}
Due to the assumption $F_u F_{u_{1,0}} F_{u_{0,1}} \ne 0$, we can rewrite \eqref{uij} in any of the following forms
\begin{eqnarray*}
&&u_{-1,-1}=\bar{F}(u,u_{-1,0},u_{0,-1}),  \label{umm}\\
&&u_{1,-1}=\hat{F}(u,u_{1,0},u_{0,-1}), \label{upm} \\
&&u_{-1,1}=\tilde{F}(u,u_{-1,0},u_{0,1}).  \label{ump}
\end{eqnarray*}
This allows us to express any `mixed shift' $u_{m,n}$, $n m \ne 0$, in terms of \emph{dynamical variables} $u_{k,0}$, $u_{0,l}$. The notation $g[u]$ indicates that the function $g$ depends on a finite number of the dynamical variables, while $f[i,j,u]$ designates that $f$ may explicitly depend on $i$, $j$ and a finite set of the dynamical variables which is the same for all $i$ and $j$. All functions are assumed to be analytical.

By $T_i$ and $T_j$ we denote the operators of the forward shifts in $i$ and $j$ by virtue of the equation~\eqref{uij}, while 
$T_i^{-1}$ and $T_j^{-1}$ denote the inverse (backward) shift operators. A shift operator with a superscript $k$ designates the $k$-fold application of this operator, and we let any operator with the zero superscript be equal to the identity mapping. The shift operators are defined by the following rules:
\[ 
T_i^k(f(i,j,a,b,\dots))=f(i+k,j,T_i^k(a),T_i^k(b),\dots) \quad \Rightarrow \quad T_i^k(u_{m,0})=u_{m+k,0}, 
\]
\[ 
T_j^k (f(i,j,a,b,\dots))=f(i,j+k,T_j^k (a),T_j^k (b),\dots) \quad \Rightarrow \quad T_j^k (u_{0,m})=u_{0,m+k},
\]
\[
T_i(u_{0,n})=T_j^{n-1}(F), \qquad T_i(u_{0,-n})=T_j^{1-n}(\hat{F}), 
\]
\[
T_j(u_{n,0})=T_i^{n-1}(F), \qquad T_j(u_{-n,0})=T_i^{1-n}(\tilde{F}), 
\]
\[
T_i^{-1}(u_{0,n})=T_j^{n-1}(\tilde{F}), \qquad T_i^{-1}(u_{0,-n})=T_j^{1-n}(\bar{F}), 
\]
\[
T_j^{-1}(u_{n,0})=T_i^{n-1}(\hat{F}), \qquad T_j^{-1}(u_{-n,0})=T_i^{1-n}(\bar{F}) 
\]
for any function $f$ and any integers $k$, $m$ and $n>0$. 
\begin{definition} A function $f[i,j,u]$ is called a \emph{symmetry} of equation
\eqref{uij} if the relation $L(f)=0$ holds for all $i$, $j$, where
\begin{equation}\label{lopd}
L = T_i T_j - F_{u_{1,0}} T_i - F_{u_{0,1}} T_j - F_u.
\end{equation}
Operators $R=\sum_{k=0}^{r} \lambda_k[u] T_i^k$ and $\bar{R}=\sum_{k=0}^{\bar{r}} \bar{\lambda}_k[u] T_j^k$  are said to be $i$- and $j$-\emph{symmetry drivers}, respectively, if $\lambda_{r} \bar{\lambda}_{\bar{r}} \ne 0$, $r, \bar{r} \ge 0$ and $R(\xi)$, $\bar{R}(\eta)$ are symmetries of \eqref{uij} for any $\xi[i,j,u] \in \ker (T_j-1)$ and any $\eta[i,j,u] \in \ker (T_i-1)$.  
\end{definition} 
\begin{definition} An equation of the form \eqref{uij} is called \emph{Darboux integrable} if there exist functions $I(u_{\ell,0},u_{\ell+1,0},\dots,u_{\ell +m,0})$ and $J(u_{0,\bar{\ell}},\dots,u_{0,\bar{\ell} + n})$ such that $I_{u_{\ell,0}} I_{u_{\ell+m,0}} \ne 0$, $J_{u_{0,\bar{\ell}}} J_{u_{0,\bar{\ell}+n}} \ne 0$ and the equalities $T_j(I)=I$, $T_i(J)=J$ hold. The functions $I$ and $J$ are respectively called an $i$-\emph{integral of order} $m$ and a $j$-\emph{integral of order} $n$ for the equation \eqref{uij}.

We say that operators $\mathcal{I}=\sum_{k=0}^{m} \mu_k[u] T_i^k$ and $\mathcal{J}=\sum_{k=0}^{n} \beta_k[u] T_j^k$ are \emph{formal} $i$- and $j$-\emph{integrals}, respectively, if $\mu_m \ne 0$, $\beta_n \ne 0$, $m, n > 0$ and the operator identities $(T_j - 1) \mathcal{I} = \sum_{k=0}^{m-1} \nu_k[u] T_i^k L$, $(T_i-1) \mathcal{J} = \sum_{k=0}^{n-1} \gamma_k[u] T_j^k L$ hold for some functions $\nu_k[u]$, $\gamma_k[u]$. 
\end{definition}

Since $T_i^{-\ell}$ and $T_j^{-\bar{\ell}}$ respectively map $i$- and $j$-integrals into $i$- and $j$-integrals again, we can assume $\ell = \bar{\ell} = 0$ without loss of generality. Under this assumption, $I_{*}= \sum_{k=0}^m I_{u_{k,0}} T_i^k$ and  $J_{*}= \sum_{k=0}^n J_{u_{0,k}} T_j^k$ are formal $i$- and $j$-integrals by Lemma~3 in \cite{Stn}. 

According to \cite{AdS}, equation \eqref{uij} admits both $i$- and $j$-symmetry drivers if this equation possesses both $i$- and $j$-integrals. The proof of this statement was omitted in \cite{AdS} because it is very similar to the proof of the analogous statement for the semi-discrete equations~\eqref{uix}. We give this proof below for the reader convenience and to demonstrate that the proof remains valid for formal integrals too. For further reasonings, we again need to introduce Laplace invariants.

As in the differential-difference case, the operator \eqref{lopd} can be represented in the form
\[ L = \left(T_i-F_{u_{0,1}}\right) \left(T_j-T_i^{-1}\left(F_{u_{1,0}}\right)\right)-H_{0} = 
\left(T_j-F_{u_{1,0}}\right) \left(T_i-T_j^{-1}\left(F_{u_{0,1}}\right)\right)-G_{0}, \]
where $H_0= F_u + F_{u_{0,1}} T_i^{-1}\left(F_{u_{1,0}}\right)$, $G_0= F_u + F_{u_{1,0}} T_j^{-1}\left(F_{u_{0,1}}\right)$. Using $L_0:=L$ as a starting term and the operator equality
\begin{equation}\label{xl}
(T_j-b_{k+1}) L_{k} = L_{k+1} (T_j - T_i^{-1}(b_k))
\end{equation}
as a defining relation for the sequence of the operators $L_k=(T_i - T_j^k(F_{u_{0,1}}))(T_j-T_i^{-1}(b_k)) - H_k=$ $=(T_j-b_k) (T_i - T_j^{k-1}(F_{u_{0,1}})) - T_j(H_{k-1})$, we obtain
\[ b_{k+1}= \frac{T_i^{-1}(b_k) T_j (H_k)}{H_k},\quad b_0=F_{u_{1,0}}, \quad H_{k+1}= T_j(H_k) - T_j^k\left(F_{u_{0,1}}\right) b_{k+1} + T_j^{k+1}\left(F_{u_{0,1}}\right) T_i^{-1}(b_{k+1}). \]
The functions $H_k$ are called \emph{Laplace $j$-invariants} of \eqref{uij}. Laplace $i$-invariants $G_k$ are defined analogously. It is convenient for further reasonings to introduce the difference operators
\begin{equation}\label{Bj} 
B_{-1}=\bar{B}_0=1, \qquad B_k= (T_j - T_i^{-1}(b_k)) B_{k-1}, \quad \bar{B}_{k+1}= (T_j - b_{k+1}) \bar{B}_k, \quad k \ge 0.
\end{equation}

\begin{proposition}[\cite{AdS}]\label{p2}
Let equation~\eqref{uij} admits a formal $j$-integral $\mathcal{J}=\sum_{k=0}^{n} \beta_k[u] T_j^k$. Then $H_p=0$ for some $p < n$ and $T_j^{1-n}(\beta_n) \in \ker(T_i - F_{u_{0,1}}^{-1})$. If, in addition, $p=n-1$, then $\mathcal{J}= \beta_n B_p$, where $B_p$ is defined by \eqref{Bj}.
\end{proposition}
\begin{proof} Equation~\eqref{xl} implies $L_k B_{k-1} = \bar{B}_k L$ and
\begin{equation}\label{struc}
T_i B_k = L_k B_{k-1} + T_j^k(F_{u_{0,1}}) B_k +  H_k B_{k-1} = T_j^k(F_{u_{0,1}}) B_k +  H_k B_{k-1} + \dots,\quad k \ge 0,
\end{equation}
where the dots denote terms of the form $\zeta_\ell [u] T_j^{\ell} L$. If $H_k \ne 0$ for all $k<n-1$, then $\mathcal{J}$ can be rewritten as $\sum_{k=0}^{n} \tilde{\beta}_k[u] B_{k-1}$, $\tilde{\beta}_n = \beta_n$. Substituting this into the defining relation of formal $j$-integrals, taking \eqref{struc} into account and collecting the coefficients at $T_i$ and $B_{k-1}$, we obtain 
\[ T_i(\tilde{\beta}_0) = 0, \quad T_i(\tilde{\beta}_1) H_0=0, \qquad T_i(\tilde{\beta}_{k+1}) H_k = (1 -T_j^{k-1}(F_{u_{0,1}}) T_i)(\tilde{\beta}_k), \quad 1 \le k < n.  \]
The above relations imply that $\tilde{\beta}_k=0$ for all $k \le n$ if $H_k \ne 0$ for all $k<n$. But this contradicts the condition $\beta_n \ne 0$ and, hence, $H_p=0$ for some $p<n$. If $p=n-1$, then $\tilde{\beta}_k=0$ for all $k \le n-1$ and $\mathcal{J}= \beta_n B_p$. Since $T_i T_j^k = T_j^{k-1} L + T_j^{k-1} (F_{u_{0,1}}) T_j^k\,+$ terms without $T_j^k$, we obtain $(T_j^{n-1}(F_{u_{0,1}}) T_i - 1)(\beta_n)  = 0$ by collecting coefficients of $T_j^n$ in the defining relation for formal $j$-integrals. 
\end{proof}
The converse statement is also true. 
\begin{proposition}\label{t2} 
Let a Laplace $j$-invariant $H_p$ of equation \eqref{uij} be equal to zero and there exist a non-zero function $\theta[u] \in \ker(T_i - F_{u_{0,1}}^{-1})$. Then $T_j^p(\theta) B_p$, where $B_p$ is defined by \eqref{Bj}, is a formal $j$-integral of \eqref{uij}.
\end{proposition}
\begin{proof}
The equalities $H_p=0$ and \eqref{xl} imply
\begin{equation}\label{pr1}
(T_i-T_j^p(F_{u_{0,1}}))B_p = L_p B_{p-1} = \bar{B}_p L.
\end{equation}
Since $T_i(\theta)=F_{u_{0,1}}^{-1} \theta$, we have  
\[ (T_i-1) T_j^p(\theta) = T_j^p(\theta) (T_j^p(F^{-1}_{u_{0,1}}) T_i - 1)= T_j^p(\theta F^{-1}_{u_{0,1}}) (T_i- T_j^p(F_{u_{0,1}})). \]
Multiplying \eqref{pr1} by $T_j^p( \theta F^{-1}_{u_{0,1}})$, we therefore obtain 
$(T_i-1) T_j^p(\theta) B_p = T_j^p(\theta F^{-1}_{u_{0,1}}) \bar{B}_p L$. 
\end{proof}

\begin{proposition}[\cite{AdS}]\label{p3}
Let equation~\eqref{uij} admit a non-zero function $\vartheta[u] \in \ker (T_j - F_{u_{1,0}}^{-1})$ and $H_p=0$ for some $p \ge 0$. Then this equation possesses the $i$-symmetry driver
\[ R = \begin{cases}
\frac{1}{H_0} \left(T_i - F_{u_{0,1}}\right) \dots \frac{1}{H_{p-1}} \left(T_i - T_j^{p-1} (F_{u_{0,1}})\right) \frac{T_i^{-1} (H_{p-1}) \dots T_i^{-p}(H_0)}{T_i^{-(p+1)}(\vartheta)} &\text{if $p>0$,} \\
T_i^{-1}(\vartheta^{-1}) &\text{if $p=0$.}
\end{cases} \]
\end{proposition}
\begin{proof} It is easy to check that $L_k H_k^{-1} (T_i - T_j^k (F_{u_{0,1}}))= (T_i - T_j^k (F_{u_{0,1}})) T_j(H_k^{-1}) L_{k+1}$ for all $k \ge 0$ and $H_k^{-1} (T_i - T_j^k (F_{u_{0,1}}))$ therefore maps $\ker L_{k+1}$ into $\ker L_k$. Any element of $\ker (T_j-T_i^{-1}(b_p))$ belongs to $\ker L_p$ if $H_p = 0$. Taking  the equalities $(T_j-b_k) H_{k-1} = T_j(H_{k-1}) (T_j - T_j^{-1}(b_{k-1}))$ and $(T_j - b_0) \vartheta^{-1} =\vartheta^{-1} b_0 (T_j-1)$ into account, we obtain
\[ \left(T_j - T_i^{-1} (b_p)\right) \frac{T_i^{-1} (H_{p-1}) \dots T_i^{-p}(H_0)}{T_i^{-(p+1)}(\vartheta)} =  \frac{T_j \left(T_i^{-1} (H_{p-1}) \dots T_i^{-p}(H_0)\right) T_i^{-(p+1)}(b_0)}{T_i^{-(p+1)}(\vartheta)} (T_j-1).\]  
Thus, the multiplication by $T_i^{-1} (H_{p-1}) \dots T_i^{-p}(H_0) T_i^{-(p+1)}(\vartheta^{-1})$  (by $T_i^{-1}(\vartheta^{-1})$ if $p=0$) maps $\ker (T_j-1)$ into $\ker (T_j-T_i^{-1}(b_p)) \subset \ker L_p$ and $R$ is a symmetry driver of \eqref{uij}.
\end{proof}
Again, the converse statement to Proposition~\ref{p3} is also true. 
\begin{proposition}\label{t3}
If equation \eqref{uij} admits an $i$-symmetry driver $R=\sum_{k=0}^{r} \lambda_k[u] T_i^k$, then $H_p=0$ for some $p \le r$ and $T_i(\lambda_r^{-1}) \in \ker (T_j - F_{u_{1,0}}^{-1})$. 
\end{proposition}
\begin{proof}
Collecting the coefficients at $f(i+k)$, $k =\overline{0, r + 1}$, in the equality $L(R(f(i))) = 0$ and taking
the arbitrariness of $f$ into account, we obtain the following chain of the relations
\begin{eqnarray*} 
&&T_i\left(B_0(\lambda_r)\right) =0, \\ 
&&T_i\left(B_0(\lambda_{k-1})\right) = (F_{u_{0,1}} T_j + F_u) (\lambda_k), \qquad  1 \le k \le r, \\ 
&&(F_{u_{0,1}} T_j + F_u) (\lambda_0) = 0. 
\end{eqnarray*}
It is easy to check that $F_{u_{0,1}} T_j + F_u = T_i B_0 - L$ and the above chain can be rewritten as
\begin{eqnarray}
&& B_0(\lambda _r)=0, \label{B0} \\
&&T_i\left(B_0(\lambda _{k-1}-\lambda_k)\right) + L(\lambda_k)= 0, \qquad 1 \le k \le r, \label{Brj} \\
&&T_i\left(B_0(\lambda _0)\right) - L(\lambda_0) = 0. \label{Br} 
\end{eqnarray}
Applying the operator $\bar{B}_{r-k+1}$ to \eqref{Brj} and taking \eqref{xl} into account, we obtain
\[ T_i\left(B_{r-k+1}(\lambda_{k-1} - \lambda_k)\right) + L_{r-k+1} \left(B_{r-k} (\lambda_k)\right) = 0, \qquad 1 \le k \le r,\]
where $\bar{B}_k$ and $B_k$ are defined by \eqref{Bj}. Thus, $B_{r-k+1}(\lambda_{k-1})=0$ if $B_{r-k}(\lambda_k) = 0$. This and \eqref{B0} imply $B_{r-k}(\lambda_k) = 0$ for all $k$ from $r$ to $0$. And the equality $B_{r-k}(\lambda_k) = 0$ gives us the relations 
\begin{equation}\label{hh} 
\bar{B}_{r-k} L (\lambda_k) = L_{r-k} \left(B_{r-k-1} (\lambda_k)\right) = - H_{r-k} B_{r-k-1} (\lambda_k).
\end{equation}

Now let us apply the operators $\bar{B}_r$ and $\bar{B}_{r-k}$ to the equalities \eqref{Br} and \eqref{Brj}, respectively. Taking \eqref{hh} into account and using the notations $\Lambda_k=B_{r-k-1} (\lambda_k)$, we obtain 
\[ H_r \Lambda_0 = 0, \qquad H_{r-k} \Lambda_k = T_i(\Lambda_{k-1}), \quad 1 \le k  \le r. \]
The above chain of the relations means that $\Lambda_k=0$ for all $k=\overline{0,r}$ if $H_k=0$ for all $k \le r$. But this contradicts the condition $\Lambda_r=\lambda_r \ne 0$ in the definition of $i$-symmetry drivers.
\end{proof}
Propositions~\ref{p2}--\ref{t3} (and their `symmetrical' versions obtained by interchanging $i \leftrightarrow j$) together prove the following statement.
\begin{theorem}
Equation \eqref{uij} admits both formal $i$-integrals and formal $j$-integrals if and only if this equation possesses both $i$- and $j$-symmetry drivers.
\end{theorem}

We can use formal integrals for obtaining `genuine' integrals, all the ways described at the end of Section~\ref{s2} are applicable for this purpose in the pure discrete case too. As an example, let us consider the equation
\begin{equation}\label{dliuv}
u_{1,1} = \frac{(u_{1,0}-1)(u_{0,1}-1)}{u}
\end{equation} 
from \cite{Hirota}. Since $H_1=G_1=0$ for this equation, we can calculate its formal $j$-integral
\[ \mathcal{J} = \frac{u_{0,1}+u-1}{u_{0,1}(u_{0,1}-1)} T_j^2 - \left( \frac{u_{0,2}+u_{0,1}-1}{u_{0,1}^2(u_{0,1}-1)}(u-1) + \frac{u_{0,2} (u_{0,1}+u-1)}{u_{0,1}(u_{0,1}-1)^2} \right) T_j  + \frac{u_{0,2}+u_{0,1}-1}{u_{0,1}(u_{0,1}-1)}\]
by using Proposition~\ref{t2}. The formal $i$-integral of \eqref{dliuv} is defined by the same formula (up to interchanging $i \leftrightarrow j$). Solving the equation $J_*=\mathcal{J}$, we find the function 
\[ J = \left( \frac{u_{0,2}}{u_{0,1} - 1} +1 \right) \left( \frac{u-1}{u_{0,1}} +1 \right),\]
which coincides (up to a point transformation) with the $j$-integral of \eqref{dliuv} obtained in \cite{AdS}.  By the version of Proposition~\ref{p3} for the case $G_p=0$, the equation~\eqref{dliuv} has the $j$-symmetry driver
\[ \bar{R} = \frac{u_{0,-1} (u_{0,-1} - 1)}{u_{0,-1}+u_{0,-2}-1} T_j - \frac{u (u-1)}{u_{0,-1} (u_{0,-1}-1)} \cdot \frac{u_{0,-2} (u_{0,-2} - 1)}{u_{0,-2}+u_{0,-3}-1}.  \] 
The coefficients of the composition $\mathcal{J} \bar{R}$ also give us $j$-integrals:
\[ \mathcal{J} \bar{R} = T_j^3 + \frac{J(1-J)}{T_j^{-1}(J)} T_j^2 + \frac{(J-1)T_j^{-1}(J)}{T_j^{-2}(J)} T_j - \frac{J T_j^{-2}(J)}{T_j^{-1}(J)T_j^{-3}(J)}. \]
As in the semi-discrete case, it seems to be plausible that `genuine' integrals can always be obtained from formal integrals by at least one of the methods discussed in the last four paragraphs of Section~\ref{s2}. Note that one of these methods was successfully used to construct integrals even for an non-autonomous quad-graph equation in \cite{GYumj}.


\end{document}